\documentclass [11pt,a4paper]{article}

\usepackage{booktabs}
\usepackage{rotating}

\usepackage{epstopdf}
\usepackage{subfigure}

\usepackage{times}
\usepackage[sans]{dsfont}
\usepackage{amscd}
\usepackage{underscore}
 \usepackage{xcolor}

\usepackage{mathrsfs}

\usepackage[leqno]{amsmath} 
\usepackage{amsfonts}
\usepackage{amsthm}
\usepackage{latexsym}
\usepackage{amssymb}
\usepackage[scaled=0.9]{helvet}
\usepackage{mathptmx}

\usepackage{graphicx}  
\usepackage{ifthen}

\newcommand{\rio}{\mathds{R}}    
\newcommand{\nio}{\mathds{N}}    

\newcommand{\emdf}{\mathds{F}}

\newcommand{\plim}{\stackrel{\proba}{\rightarrow}}

\newcommand{\ind}{\mathbf{I}}
\newcommand{\proba}{\mathsf{P}}

\newcommand{\mean}{\mathsf{E}}

\newcommand{\sd}{\mathsf{sd}}

\newcommand{\ud}{\mathrm{d}}

\newcommand{\ssd}{\succcurlyeq_{\text{SSD}}}

\newcommand{\siga}{\mathscr{F}} 
\newcommand{\slaw}{\stackrel{\text{law}}{=}}
\providecommand{\norm}[1]{\lVert#1\rVert}

\newcommand{\VaR}{\mathrm{VaR}}

\newcommand{\AVaR}{\mathrm{AV@R}}

\newcommand{\CR}{\mathrm{CR}}
\newcommand{\SR}{\mathrm{SR}}

\newtheorem{theorem}{Theorem}
\newtheorem{lemma}{Lemma}
\newtheorem{prop}{Proposition}

\newtheorem{remark}{Remark}

\newcommand{\defin}{\fontfamily{ptm}\fontseries{b}
\fontsize{12}{12}\selectfont}

\newtheoremstyle{mystyle1}{}{}{\normalfont}%
{0cm}{\defin}{.}{0.2cm}{}

\theoremstyle{mystyle1}

\newtheorem{defi}{Definition}

\newcommand{\rmk}{\fontfamily{ptm}\fontseries{m}\fontshape{sc}
\fontsize{12}{12}\selectfont}

\newtheoremstyle{mystyle2}{}{}{\normalfont}%
{0cm}{\rmk}{.}{0.2cm}{}

\theoremstyle{mystyle2}

\begin{document}

\title{Acceptability Indices of Performance for Bounded C\`{a}dl\`{a}g Processes}
\author{Christos E. Kountzakis \thanks{%
Department of Mathematics, University of the Aegean, 83200 Samos Greece, email: {\tt chr_koun@aegean.gr} },
Damiano Rossello \thanks{%
\textit{Corresponding Author}: Department of Economics and
Business, University of Catania, 95129 Catania,
Italy, email: {\tt rossello@unict.it} }
}

\maketitle

\begin{abstract}
\noindent Indices of acceptability are well suited to frame the axiomatic features of many performance measures, associated to terminal random cash flows. We extend this notion to classes of c\`{a}dl\`{a}g processes modelling cash flows over a fixed investment horizon. We provide a representation result for bounded paths. We suggest an acceptability index based both on the static Average Value-at-Risk functional and the running minimum of the paths, which eventually represents a RAROC-type model. Some numerical comparisons clarify the magnitude of performance evaluation for processes.\\


\end{abstract}


\section{Introduction}

 \noindent A financial performance aims to evaluate the return characteristics of a given investment portfolio, and to fulfill the risk and asset allocation constraints provided by investors. The resulting measurement can be used to judge the quality of managerial skillfulness, since fund managers viewed as competitors should be able to process those piece of information not reflected by market prices, then providing an actual value-added service. The balance between reward and risk is condensed into a \emph{performance measure}
 as the popular Sharpe ratio (SR), which can be used to rank investment portfolios according to these two characteristics. Other performance measures are alternative to SR, accounting for stylized facts about financial returns such as asymmetrical and fat-tailed distributions. See
 \cite{AmLeSou:2003} for a textbook treatment of several performance measures.

 \noindent The unified mathematical framework of Cherny \& Madan (2009) deals with some desirable features many performance measures share. Given a set of (at least) four properties, they identify \emph{indices of acceptability} of nonnegative financial performance for a given terminal random cash flow. The basic idea is intimately linked to the well known analysis in Atzner et al. (1999) of coherent risk measure and their acceptance sets, and the subsequent analysis in Carr et al. (2001) of arbitrage pricing and acceptance hedging: There is a continuum of degrees of acceptability of a position, varying with different levels of stressed scenarios supporting positive expectations of cumulative terminal cash flow. Thanks to the duality between coherent risk measures and performance measures, a class of equivalent probability measures, or the corresponding set of their Radon-Nykodim derivatives, give the acceptability of a trade's cash flow. In other words, an index of acceptability derived from a coherent risk measure must be proportional to the amount of stress tolerated and must yield nonnegative values exceeding floors, so that nonnegative expected cash flows are attained. This fits well the regulatory-capital requirements and the pertaining economic-capital modelling used in practice, for example by banks undertaking ex ante improvements in business-performance tracking through the use of risk adjusted return on capital (RAROC), as the ratio of expected final return to the economic capital measured by a coherent risk measure or by the Value-at-Risk (VaR). Not surprisingly, an acceptability index (AI) may be expressed just in ratio form, as the forerunner static SR.

\noindent On the other hand, the industry of fund management claims the use of performance measures such as the Calmar ratio (CR) designed to account for the risk associated to cash flow resulting from the whole investment's path over a fixed horizon. If this is the case, SR is redefined by set the reward measure in the numerator equal to the expected terminal cash flow, and using the expected maximum drawdown over the horizon as a risk measure in the denominator. The risk-adjusted return then takes into account the future evolution of the market value of the position, not just the terminal one. Albeit this kind of performance measure is widespread among practitioners, it cannot be directly placed into the realm of static AIs studied in \cite{ChernyMad:2009}. Static risk measures does not embed the cash flow's path experienced over time. On the other hand, dynamic risk measures have been developed to account for this. First, static risk measures are turned into conditional ones to account for the information available at the risk assessment. Then, on some filtered probability space risk assessment is updated as time elapses and new information arrives, so that a sequence of conditional risk measures depicts a dynamical framework based on different notions of \emph{time consistency}, see \cite{AccPenn11} and the references therein for a detailed review. It is possible to turn things around, and define a coherent monetary risk measure as in Cheridito et al. (2004), yielding a numerical evaluation rather than a sequence of conditional risk measures (random variables). For the special case of a finite sequence of adapted cash flows see \cite[Section 3.2]{PflugRom07}.

\noindent In the present paper we provide a framework for AIs of performance put in duality with coherent monetary risk measures for bounded c\`{a}dl\`{a}g processes. To recover the information lost in the static setting, one records all possible stressed scenarios during the holding period of a financed position, which is represented by a path rather than a random variable. However, we do not develop a dynamic setup to process information, rather we give a representation based on a static index acting on processes which preserves
the main features of an AI: Quasi-concavity, positive homogeneity, monotonicity and Fatou continuity. Eventually, this includes the one-period AI as a particular case. Furthermore, we propose an example of AI for processes related to the one-time step RAROC. Our contribution is similar to that of Bielecki et al. (2014), or the more recent Bielecki et al. (2017), which define coherent risk and AIs for paths in a newly proposed dynamical setting. These authors studied the duality between \emph{dynamic coherent risk measures} $\rho:\{0, \ldots,T\} \times \mathscr{D} \times \Omega \to \rio$  and \emph{dynamic coherent acceptability indices} $\rho:\{0, \ldots,T\} \times \mathscr{D} \times \Omega \to [0, \infty],$ where $\{0, \ldots,T\}$ is a set of dates and $\mathscr{D}$ is a set of (adapted) real-valued stochastic processes modelling cash flows. These authors impose two additional properties to $\alpha$ and $\rho$ called \emph{independence of the past} and \emph{dynamic consistency}. Nevertheless, this framework cannot deal with performance measures such as the aforementioned CR. Our contribution enable us to overcome this limitation, and to deal with path-transformations like taking the maximum drawdown or the running minimum of a cash flow's path. Eventually information about time resolution of an investment process (from the spatial viewpoint) is not lost. On the generalization of the AIs to the continuous-time setting see Biagini and Bion-Nadal (2015).

\noindent The paper proceeds as follows. Section \ref{sec:AIoP} introduces the essential toolkit for treating some classes of c\`{a}dl\`{a}g processes as models of total cumulative cash flows evolving within a finite horizon, and defines the AIs of performance in this extended framework. The duality concerning such classes are briefly reviewed, together with additional results on their lattice structure used in proving the main representation of this paper. Section \ref{sec:RepreTH} is devoted to the generalization of the static AI of \cite{ChernyMad:2009} from the domain $L^{\infty}$ to  the collection of bounded c\`{a}dl\`{a}g processes. The new AI is obtained in a straightforward manner, by properly combining the contributions \cite{ChernyMad:2009,CheriEtAl:2004}. The multi-period analogue of the static system of supporting kernels is depicted in terms of bi-variate processes reproducing the Radon-Nikodym derivatives corresponding to a static system of kernels. Section \ref{sec:Further-Pr1} deals with the analysis of second order stochastic dominance compatible with the extended AI. Section \ref{sec:Arbitr-Expec} studies its arbitrage and expectation consistency. Section \ref{sec:Examples} contains the main example of an AI for processes based on the one-period RAROC, as the ratio of the expected terminal cumulative cash flow to the Average Value-at-Risk of the cash flow's running minimum. Section \ref{sec:Examples} provides numerical comparisons of simulated values in order to appreciate the magnitude of  AIs and other performance measures. Section \ref{sec:Concl} contains some concluding remarks.

\section{Notations and Preliminary Results}\label{sec:AIoP}

In this paper we model the whole evolution of financial outcomes over a finite time-interval rather than the terminal cumulative cash flow typically handled in performance analysis. Here $X=(X_t)_{t \in [0,T]}$ is the stochastic process modeling the random cash flow resulting from dynamic trading over the investment horizon $[0,T],$ where $T \geqslant 0.$ We are given a filtered probability space $(\Omega,\siga, (\siga_t)_{t \in [0,T]}, \proba)$ satisfying the usual conditions, i.e., the basis space $(\Omega,\siga,\proba)$ is complete, the filtration $(\siga_t)_{t \in [0,T]}$ is right-continuous, and the initial information $\siga_0$ contains all the $\proba$-null events of $\siga.$ Almost surely equal random variables are identified as well as indistinguishable processes on the filtered space, $X_t(\omega)=Y_t(\omega)$ for almost all $\omega \in \Omega$ and all $t \geqslant 0.$ Comparisons among processes are understood in the latter sense. For example, $X \geqslant Y$ means that $X_t$ is greater than or equal to $Y_t,$ for all dates $t$ and for almost all $\omega.$ As usual we set $L^p:=L^p(\Omega, \siga, \proba)$ and following \cite{CheriEtAl:2004,DellachMeyer:1982} we denote $\mathscr{R}^0$ the vector space of (the $\proba$-a.s. and for every $t \in [0,T]$ equivalence classes of) c\`{a}dl\`{a}g processes that are adapted to the filtration. For the characterization we develop in Section \ref{sec:RepreTH}, the model $X$ of a (discounted) cash flow evolving within the horizon is that of a bounded c\`{a}dl\`{a}g process belonging to the stricter class $\mathscr{R}^{\infty},$ i.e. $X^{*}:=\sup_{t \in [0,T]} |X_t| \in L^{\infty}.$
This is a Banach space equipped with the norm $\left\| X\right\|_{\mathscr{R}^{\infty}}:=\left\| X^{*} \right\|_{\infty},$  where $\left\| \cdot \right\|_{\infty}$ is the usual norm on $L^{\infty}.$
\begin{defi}\label{defi:alpha-phi}
    A performance measure $\alpha : \mathscr{R}^{\infty} \to \rio$ is an AI for processes if it satisfies the following properties:
    \begin{itemize}
      \item[(1)] Acceptable cash flows at a level $x \geqslant 0$ form a convex above-level set
                    $$\mathscr{A}_x := \left\{ X \in \mathscr{R}^{\infty} \big| \alpha(X) \geqslant x\right\}.$$
                 In the current context, this is a family describing acceptability for any level $x.$ The convexity of any
                 $\mathscr{A}_x$ is equivalent to the quasi-concavity of $\alpha,$ namely $\alpha(\lambda X + (1-\lambda)Y) \geqslant x$
                 for any $\lambda \in [0,1]$ provided $X,Y$ are such that $\alpha(X) \geqslant x$ and $\alpha(Y) \geqslant x.$
                 Taking $x = \min\{\alpha(X),\alpha(Y)\}$ quasi-concavity implies that a diversified position performs better than its components.
      \item[(2)] Acceptable cash flows are valued monotonically,
                    $$X \leqslant Y \Rightarrow \alpha(X) \leqslant \alpha(Y),$$
                 thus $\alpha$ is an increasing map and $Y$ is at least as acceptable as $X.$
      \item[(3)] The acceptance set $\mathscr{A}_x$ is required to be a convex cone, because $\alpha$ is not meant to be an investment criterion
                 but rather it measures to what extent moving away from marginal trades supporting the random cash flow $X$ results in a new
                 \emph{investment direction} based on alternative pricing kernels. Hence
                    $$\alpha(\lambda X)= \alpha(X), \quad \text{for}\;\; \lambda >0,$$
                 i.e. \emph{scale invariance} is required and the performance of an investment should not depend upon the initial endowment.
                 In other words, $\lambda X$ is based on a trade in the same direction of $X,$ and then it has the same level of acceptance.
      \item[(4)] The acceptability functional is required to be upper Fatou-continuous for $X \in \mathscr{R}^{\infty}$,
                        $$\limsup_{n \to \infty} \alpha(X_n) \leqslant \alpha(X),$$
                 and for every bounded sequence $(X^n)_{n \in \nio} \subset \mathscr{R}^{\infty}$ of paths such that $(X^n)^{*}$ converges in probability to $X^{*},$ i.e. $(X^n -X)^{*} \plim 0.$ This implies that $\alpha(X) \geqslant x$ provided that $\alpha(X^n) \geqslant x$
                 for every $n \in \nio$ and $x \geqslant 0.$ The bounded sequence can be taken as $\left\| X^n \right\|_{\mathscr{R}^{\infty}} \leqslant 1$
                 uniformly.
    \end{itemize}
\end{defi}
\noindent To obtain the representation results of Section \ref{sec:RepreTH}, we need the duality relations concerning spaces of c\`{a}dl\`{a}g processes, as well studied more recently in \cite{CheriEtAl:2004} and classically in \cite{DellachMeyer:1982}. We also need the lattice structure of some of such spaces which enforces the involved duality. Recall that a Riesz space is an ordered vector space which is also a lattice, where the norm and the absolute value are in general different. If the norm is in addition monotone in the absolute value of a vector, then it is a lattice norm and completeness entails a Banach lattice. Two special classes of Banach lattices are the AM-spaces and the AL-spaces, based on the following definitions of lattice norms:
\begin{itemize}
  \item $M$-norm, whenever $X,Y \geqslant 0$ implies $\norm{\sup\{X,Y\}}=\max \{\norm{X},\norm{Y}\};$
  \item $L$-norm, whenever $X,Y \geqslant 0$ implies $\norm{X+Y}=\max \{\norm{X},\norm{Y}\}.$
\end{itemize}
Thus, a norm complete Riesz space equipped with an $M$-norm is an AM-space, while a norm complete Riesz space equipped with an $L$-norm is an AL-space. A useful result states that a Banach lattice is an AM-space (resp. an AL-space) if and only if its dual is an AL-space (resp. an AM-space), see \cite[Chs 8, 9]{AB:2006} for more details. The class $\mathscr{R}^p,$ for $p \in [1,\infty],$ generalizes $\mathscr{R}^{\infty}$
to those c\`{a}dl\`{a}g processes such that $X^* \in L^p_+.$ This is also a Banach space with norm $\left\| X\right\|_{\mathscr{R}^p}:=\left\| X^{*} \right\|_{p},$ see Appendix \ref{appx} for a brief review.
\begin{prop}\label{prop:partial-order-R-inf}
$\mathscr{R}^{p}$ is a Banach lattice, for $p \in [1,\infty].$ Moreover, $\mathscr{R}^{\infty}$ is an AM-space with order unit.
\end{prop}
\begin{proof}
On $\mathscr{R}^{p}$ let us consider the partial ordering
$X \geqslant Y \Leftrightarrow X_{t}(\omega) \geqslant Y_{t}(\omega)$ for any $t \in [0,T]$, and for $\proba$-almost all $\omega \in \Omega.$ Using this partial ordering, if $X, Y \in \mathscr{R}^{p}$ and $|X| \geqslant |Y|$, which implies that  $|X_{t}(\omega)| \geqslant |Y_{t}(\omega)|$ for any $t \in [0,T]$, and for $\proba$-almost all $\omega \in \Omega,$ this also implies that the random variables $X^{*}, Y^{*} \in L^p_+$ satisfy
the inequality
    $$\| X\|_{\mathscr{R}^{p}} = \left\| X^{*} \right\|_{p} \geqslant \left\| Y^{*} \right\|_{p}
    =\| Y\|_{\mathscr{R}^{p}}.$$
Hence, $\mathcal{R}^{p}$ is a Banach lattice for $p \in [1,\infty].$ To show that  $\mathscr{R}^{\infty}$ is an AM space with unit, it suffices to prove that such an order unit is the stochastic process $\textbf{1}=(\textbf{1}_{t})_{t \in [0,T]}$, where $\textbf{1}_{t}(\omega)=1, \proba$-a.s. In order to prove it, we have to show that $\mathcal{R}^{\infty}=\cup_{n=1}^{\infty}[-n\textbf{1}, n\textbf{1}]$, where $[-n\textbf{1}, n\textbf{1}]$ denotes the following order interval of $\mathscr{R}^{\infty}$: $[-n\textbf{1}, n\textbf{1}]=\{X \in \mathcal{R}^{\infty}: n\textbf{1} \geqslant X \geqslant -n\textbf{1}\}$, with respect to the partial ordering of $\mathcal{R}^{\infty}$ defined above. The inclusion $|X_{t}| \geqslant \sup_{t \in [0,T]}|Y_{t}|$ in $L^{\infty}$, and consequently $\cup_{n=1}^{\infty}[-n\textbf{1}, n\textbf{1}] \subseteq \mathscr{R}^{\infty}$ is obvious. For the opposite inclusion, for any $X \in \mathscr{R}^{\infty}$, we may define the stochastic process $Y \in \mathscr{R}^{\infty}$, such that $Y \geqslant X$ and moreover $Y_{t}=([ |X_{t}| +1 ])\textbf{1}_{t}$.
\end{proof}
\noindent Another class of Banach lattices related to the geometry of c\`{a}dl\`{a}g processes is $\mathscr{A}^q,$ for $q \in [1,\infty],$ containing the bi-variate processes $A : [0,T] \times \Omega \to \rio^2$ such that $A=(A^{\text{pr}}_t,A^{\text{op}}_t)_{t \in [0,T]}$ has right-continuous coordinates with finite variation, $A^{\text{pr}}$ being predictable with $A^{\text{pr}}_0=0,$ while $A^{\text{op}}$ being optional and purely discontinuous, see Appendix \ref{appx}. We have $\text{Var}(A^{\text{pr}})+\text{Var}(A^{\text{op}})\in L^q,$ where $\text{Var}(\cdot)$ is he usual variation of a process. The related positive cone $\mathscr{A}^q_+$ contains those bi-variate processes $A \in \mathscr{A}^q$ such that $A^{\text{pr}},A^{\text{op}} \geqslant 0$ and increasing. The base of this cone is defined as $\mathscr{B}^q_+:= \left \{ A \in \mathscr{A}^q_+ \big| \,\langle \mathbf{1},A \rangle >0 \right \},$
where $\mathbf{1}$ is the element $X\in \mathscr{R}^p$ such that $X^*(\omega)=\mathbf{1}(\omega)=1,$ $\proba$-a.s. The partial ordering implied by $\mathscr{A}^q_+$ on $\mathscr{A}^q$ is defined as
    $$C \geqslant D \Longleftrightarrow C-D  \in \mathscr{A}^q_+,$$
and makes $\mathscr{A}^q_+$ a Banach lattices. More generally we have:
\begin{prop}\label{prop:partial-order-A}
$\mathscr{A}^{q}$ is a Banach lattice. Moreover, $\mathscr{A}^1$ it is an AL-space.
\end{prop}
\begin{proof}
The partial ordering defined on $\mathscr{A}^{q}$, is the following:
$A \geqslant B \Leftrightarrow \hat{A} \geqslant \hat{B},$ where
    $$\hat{A}:=\text{Var}(A^{\text{pr}})+\text{Var}(A^{\text{op}}) \in L^q,$$
    $$\hat{B}:=\text{Var}(B^{\text{pr}})+\text{Var}(B^{\text{op}}) \in L^q.$$
Thus, if $|A| \geqslant |B|$ this is equivalent to $\hat{|A|} \geqslant \hat{|B|}$. Thus, $\mathscr{A}^{q}$ is a Banach lattice.
This also implies $\|\hat{|A|}\|_{1}=\|\hat{A}\|_{1} \geqslant \|\hat{|B|}\|_{1}=\|\hat{B}\|_{1}$, which means that $\mathcal{A}^{1}$ is an AL-space, since $L^{1}$ is an AL-space, as well.
\end{proof}
\noindent For conjugate exponents $p=\infty$ and $q=1,$ the duality between the spaces $\mathscr{R}^{\infty}$ and $\mathscr{A}^{1}$
plays an important role in our representation of AIs for processes, namely Theorems \ref{theorem:First} and \ref{theorem:Second} in Section \ref{sec:RepreTH}, as well as other results in Sections \ref{sec:RepreTH}, \ref{sec:Arbitr-Expec} and \ref{sec:Examples}. The dual pair
$\langle \mathscr{R}^{\infty},\mathscr{A}^{1} \rangle$ is based on the dual pairing $\langle X,A \rangle$ defined on $ \mathscr{R}^{\infty} \times \mathscr{A}^{1},$ see Appendix \ref{appx}. In our main result (Theorem \ref{theorem:First}) we replace the infimum over classical non-negative expectations with respect to equivalent probability measures (viz. their Radon-Nikodym derivatives), with positive increasing dual processes with unit expected variation. Thus, for every $x \in \rio^+,$ the $x$-increasing family $(\mathscr{D}_x)_{x \in \rio_+}$ of
\cite[Theorem 1]{ChernyMad:2009} is now replaced by an $x$-increasing family $(\mathscr{Q}_{\sigma}^x)_{x \in \rio_+},$ where each $\mathscr{Q}_{\sigma}^x$ is a subset of the class $\mathscr{D}_{\sigma}$ defined in \cite{CheriEtAl:2004}, containing the bi-variate processes $A \in \mathscr{A}^1$ that are in addition nonnegative, increasing and such that $\mean[\text{Var}(A^{\text{pr}})+\text{Var}(A^{\text{op}})]=1,$ see also Appendix \ref{appx}.

\section{Basic Representation Result}\label{sec:RepreTH}

We give the analogue of  \cite[Theorem 1]{ChernyMad:2009}, to characterize an AI having a numerical value $x \in \rio_+$ such that the bounded c\`{a}dl\`{a}g path $X$ attains a positive bilinear form $\langle X,A \rangle$ (which is the analogue of the expectation in the one-period case), under each
bi-variate process $A$ from the subset $\mathscr{Q}_{\sigma} \subset \mathscr{D}_{\sigma}$ (which is the analogue of the Radon-Nykodim derivative of the absolute continuous probability measure giving the acceptability in the one-period case) corresponding
to the level $x.$ There is a one-parameter class of such sets.
\begin{theorem}\label{theorem:First}
        A map $\alpha : \mathscr{R}^{\infty} \to [0, \infty]$ is an AI for processes if and only if there exists a family
        of $x$-increasing family $(\mathscr{Q}_{\sigma}^x)_{x \in \rio_+}$ such that the representation
            \begin{equation}\label{eq:repre1}
                \alpha(X) = \sup \left\{ x \in \rio_+ \bigg| \,  \inf_{A \in \mathscr{Q}_{\sigma}^x} \langle X,A \rangle \geqslant 0 \right \}
            \end{equation}
        holds, with $\inf \varnothing = \infty$ and $\sup \varnothing = 0.$
\end{theorem}
\begin{proof}
    $(\Leftarrow)\;$ We claim that $\alpha$ defined by (\ref{eq:repre1}) satisfies the four properties of an AI. First, we check
    property (1) of Definition \ref{defi:alpha-phi} that $\alpha$ values only acceptable cash flows $X \in \mathscr{R}^{\infty}$ evolving during the
    horizon $[0,T]$ which belong to a convex level set $\mathscr{A}_x$ for any $x \in \rio_+.$ Indeed, assuming that $X$ and in addition $Y \in \mathscr{R}^{\infty}$
    entail a value of $\alpha$ which is $\geqslant x,$ then choosing $y < x$ for any bi-variate process $A \in \mathscr{Q}_{\sigma}^y$
    we have a value of the linear functional $\langle \cdot\, ,A \rangle,$ corresponding to each cash flow, which must be $\geqslant 0.$ Taking a convex combination
    for $\lambda \in [0,1]$ we then have $\langle \lambda X+ (1-\lambda)Y,A \rangle \geqslant 0,$
    and for every $A$ in the biggest class $\mathscr{Q}_{\sigma}^x$ we thus have $\inf_{A} \langle \lambda X+ (1-\lambda)Y ,A \rangle \geqslant 0$
    too. This corresponds to the greatest $x \in \rio_+$ such that $\alpha(\lambda X+ (1-\lambda)Y) \geqslant x$
    which proves the convexity of the level set $\mathscr{A}_x$ (or equivalently the quasi-concavity of the index $\alpha$). For the monotonicity,
    given two elements of $\mathscr{R}^{\infty}$ such that $X \leqslant Y,$ the stochastic integral together with the expectation operator defining the bilinear form used in (\ref{eq:repre1}) are monotone, then $\langle Y,A \rangle \geqslant \langle X,A \rangle \geqslant 0$ from which property (2) of Definition \ref{defi:alpha-phi} easily follows. Scale invariance is trivial. To show the upper Fatou-continuity, we first assume
    $\sup_{t \in [0,T]} |X_t^n| \plim  \sup_{t \in [0,T]} |X_t|$ for a bounded sequence $(X^n)_{n \in \nio}$ of elements $X^n \in \mathscr{R}^{\infty}$
    and some $X \in \mathscr{R}^{\infty}.$ This implies
        $$\left(\int_{(0,T]} X_{t^-}^n\ud A^{\text{pr}}_t + \int_{[0,T]} X_t^n \ud A^{\text{op}}_t \right)
        \plim \left(\int_{(0,T]} X_{t^-}\ud A^{\text{pr}}_t + \int_{[0,T]} X_t\ud A^{\text{op}}_t \right),$$
    thus by the Lebesgue's Dominated Convergence theorem we have
        $$  \lim_{n \to \infty} \langle X^n ,A \rangle = \langle X ,A \rangle \geqslant \limsup_{n \to \infty}
        \inf_{A \in \mathscr{Q}_{\sigma}^y} \langle X,A \rangle \geqslant 0$$
    for some $A \in \mathscr{Q}_{\sigma}^y.$ Now, for any such $y < x,$ any $A \in \mathscr{Q}_{\sigma}^y$ and any $n \in \nio$ we have $\alpha(X^n) \geqslant x$
    by construction so that the previous implies $\alpha(X) \geqslant x$ too, which is the equivalent formulation of the Fatou property
    for the AI.
\end{proof}
\noindent  Before completing the proof, we observe that representation (\ref{eq:repre1}) of Theorem \ref{theorem:First} is equivalently given by
    \begin{equation}\label{eq:AIoP-Risk}
        \alpha(X) = \sup \left\{ x \in \rio_+ \big| \, \rho_x(X)  \leqslant 0 \right \},
    \end{equation}
as pointed out in \cite{ChernyMad:2009}. Indeed, each functional on $\mathscr{R}^{\infty}$ defined by
    $$\rho_x(X):= - \inf_{A \in \mathscr{Q}_{\sigma}^x} \langle X,A \rangle, \quad x \in \rio_+$$
is by  \cite[Corollary 3.5]{CheriEtAl:2004} a coherent risk measure for processes. For $x \leqslant y,$ passing from $\mathscr{Q}_{\sigma}^x$ to the bigger set
$\mathscr{Q}_{\sigma}^y$ the value of $\rho_x(X)$ will increase to $\rho_y(X),$ for any $X \in \mathscr{R}^{\infty}.$ Then, the supremum in (\ref{eq:repre1}) of Theorem \ref{theorem:First} will increase too and obviously the equivalent representation given in (\ref{eq:AIoP-Risk}) holds true. Conversely, for a family $(\rho_x(X))_{x \in \rio_+}$
of coherent risk measures for processes $X \in \mathscr{R}^{\infty},$ which is increasing in $x$ as a map $x \mapsto \rho_x(X)$ for a fixed $X,$ any set $\mathscr{Q}_{\sigma}^x$
corresponding to a risk measure $\rho_x(X)$ in the family must be bigger anytime $x$ increases, due to the representation (\ref{eq:AIoP-Risk}).
\begin{remark}\label{remark:AccSet}
    The acceptability set $\mathscr{A}_x$ introduced in Definition \ref{defi:alpha-phi} of Section \ref{sec:AIoP} is equivalently given by
    $$\mathscr{A}_x := \left\{ X \in \mathscr{R}^{\infty} \big| \, \rho_x(X) \leqslant 0 \right\}.$$
    Thus, for every $x \in \rio_+$ we have a whole family which is clearly decreasing in $x.$ As a consequence the numerical value of an AI for processes can be recast as
        $$\alpha(X) = \sup \left\{ x \in \rio_+ \big| \, X \in \mathscr{A}_x  \right\}.$$
    There are several levels $x$ at which the performance of a trade can be measured by valuing its riskiness in an acceptable way.
\end{remark}
\noindent In order to prove the `if part' we need the following characterizations. Based on the AM-AL duality between $\mathscr{R}^{\infty}$ and $\mathscr{A}^1,$ we in addition see that for any $x \in \rio_+$ the
coherent risk measure for processes
    $$\rho_x(X)=\inf\{m \in \rio \,| \, m \cdot \textbf{1} + X \in \mathscr{A}_x\},\;\; \text{for every} \, X \in \mathscr{R}^{\infty},$$
has the following \textit{dual representation}:
    $$-\rho_{x}(X)=\inf_{\pi \in \mathscr{A}^{0}_{x}}\pi(X),$$
where $\mathscr{A}^{0}_{x}=\{f \in \mathscr{A}^{1} : f(X) \geqslant 0,$ each $X \in \mathscr{A}_{x}\},$ is the \textit{polar} set of $\mathscr{A}_{x}$ in $\mathscr{A}^{1}$. Finally, recall that a subset of a vector space is called a wedge if it is convex and it has the property that for any $x$
lying in the set we also have that $\lambda \cdot x$ belong to the same set, for every $\lambda \in \rio_+$. Putting all things together, we have that proving $(\Rightarrow)$ of Theorem \ref{theorem:First} amounts to prove
the following:
\begin{theorem}\label{theorem:Second}
For any AI $\alpha : \mathscr{R}^{\infty} \rightarrow [0,\infty],$ with the property that every level set $\mathscr{A}_{x}$ of $\alpha$ is a wedge, there exists an increasing family $(\mathscr{Q}_{\sigma}^x)_{x \in \rio_+}$ of functional sets lying in $\mathscr{A}^1$  such that $x \leqslant y$ implies $\mathscr{Q}_{\sigma}^x \subset \mathscr{Q}_{\sigma}^y$ and
        $$\alpha(X)=\sup \left\{x \in \rio_+ \bigg| \inf_{\pi \in \mathscr{Q}_{\sigma}^x}\pi(X) \geqslant 0 \right\}$$
holds.
\end{theorem}
\begin{proof}
    The level sets of $\alpha$ are $\mathscr{A}_x=\{X \in \mathscr{R}^{\infty} \, | \, \alpha(X) \geqslant x\}$. For these sets, $\mathscr{A}_y \subset \mathscr{A}_x$ holds, if $y \leqslant x$. For the equivalent polar sets $\mathscr{Q}_{\sigma}^x=\mathscr{A}^0_x$ lying in $\mathscr{A}^1,$ if $x \leqslant z$ this implies $\mathscr{Q}_{\sigma}^x \subset \mathscr{Q}_{\sigma}^z.$ Thus, any $X \in \mathscr{R}^{\infty}$ lies in some $\mathscr{A}_{x_0}.$ This implies that $\rho_{x_0}(X) \leqslant 0$, hence $-\rho_{x_0}(X) \geqslant 0$. From the equivalent dual representation of the coherent risk measure $\rho_{x_{0}}$, then
        $$\alpha(X)=\sup\{x_0 \in \rio_+  | \, -\rho_{x_0}(X) \leqslant 0\}$$
    and we are done.
\end{proof}
\noindent The family $(\mathscr{Q}_{\sigma}^x)_{x \in \rio_+}$ (viz. system of supporting kernels in \cite{ChernyMad:2009}) can be characterized as
    \begin{equation}\label{eq:repre2}
        \mathscr{Q}_{\sigma}^x = \left\{ A \in \mathscr{Q}_{\sigma} | \langle X,A \rangle \geqslant 0,\; \forall \, X \in \mathscr{R}^{\infty},\,\alpha(X) > x \geqslant 0  \right\}.
    \end{equation}
Then, we have the following maximality property:
\begin{lemma}\label{lemma:Support}
    For any AI $\alpha,$ there exists a family $(\mathscr{Q}_{\sigma}^x)_{x \in \rio_+}$ supporting the representation (\ref{eq:repre1}) and defined by
    (\ref{eq:repre2}), such that if $(\mathscr{E}_{\sigma}^x)_{x \in \rio_+}$ is a different $x$-increasing family satisfying (\ref{eq:repre1}), then it holds
    $\mathscr{E}_{\sigma}^x \subset \mathscr{Q}_{\sigma}^x$ for any $x \in \rio_+.$
\end{lemma}
\begin{proof}
    As in the proof of \cite[Proposition 1]{ChernyMad:2009} we use a squeezing argument and show that equation (\ref{eq:repre1}) can be split in two opposite inequalities.
    Let us suppose that $\alpha(X),$ defined by (\ref{eq:repre1}) and supported by some $x$-increasing family $(\mathscr{E}_{\sigma}^x)_{x \in \rio_+},$
    is strictly greater than $\sup \{ x \in \rio_+ | \,  \inf_{A \in \mathscr{Q}_{\sigma}^x} \langle X,A \rangle \geqslant 0\},$ for any process $X \in \mathscr{R}^{\infty}.$ Then
        $$\alpha(X) > y > \sup \left\{ x \in \rio_+ \bigg| \,  \inf_{A \in \mathscr{Q}_{\sigma}^x} \langle X,A \rangle \geqslant 0  \right\},$$
    for some $y \in \rio_+.$ But this implies the existence of $A \in \mathscr{Q}_{\sigma}^y$ which makes negative the bilinear form inside the supremum,
    contradicting the definition of $\mathscr{Q}_{\sigma}^y$ in (\ref{eq:repre2}). To show the reverse inequality, let us suppose that $\mathscr{E}_{\sigma}^x \supset \mathscr{Q}_{\sigma}^x.$ Then, we can find $A \in \mathscr{E}_{\sigma}^x$ which again makes negative the bilinear form and at the same time makes $\alpha(X) > x \in \rio_+,$  contradicting the definition (\ref{eq:repre1}) of $\alpha.$ Then,
        $$\alpha(X) \geqslant \sup \left\{ x \in \rio_+ \bigg| \,  \inf_{A \in \mathscr{Q}_{\sigma}^x} \langle X,A \rangle \geqslant 0 \right \}$$
    and $(\mathscr{Q}_{\sigma}^x)_{x \in \rio_+}$ is a maximal family.
\end{proof}
\noindent We state a Lemma which will be useful in the identification of typical AIs for processes, provided that a family
$(\mathscr{Q}_{\sigma}^x)_{x \in \rio_+}$ is meant to supports $\alpha$ as given by Lemma \ref{lemma:Support}.
          Recall that the space $\mathscr{Q}_{\sigma} \subset \mathscr{A}^1$ inherits the norm $\left\| \,\cdot\,\right\|_{\mathscr{A}^1}$ and then it is a Banach space.
    \begin{lemma}\label{lemma:Closedness}
    Define an AI $\alpha$ by (\ref{eq:repre1}). Let $(\mathscr{Q}_{\sigma}^x)_{x \in \rio_+}$
    be a family of convex $\left\| \,\cdot\,\right\|_{\mathscr{A}^1}$-closed subsets of $\mathscr{Q}_{\sigma}$ that are minimal with respect to intersection, i.e.
        $$\mathscr{Q}_{\sigma}^x := \cap_{y > x} \mathscr{Q}_{\sigma}^y \qquad \text{for any}\;\; x \in \rio_+.$$
    Then $(\mathscr{Q}_{\sigma}^x)_{x \in \rio_+}$ supports $\alpha$ in the representation (\ref{eq:repre1}).
\end{lemma}
\begin{proof}
    Let $(\mathscr{E}_{\sigma}^x)_{x \in \rio_+}$ be the $x$-increasing family supporting $\alpha.$ For some $x \in \rio_+,$ take
    a nonempty $\left\| \,\cdot\,\right\|_{\mathscr{A}^1}$-closed and convex set $\mathscr{Q}_{\sigma}^x \subset \mathscr{E}_{\sigma}^x.$
    This enable us to find some $y > x$ and $B \in \mathscr{E}_{\sigma}^x$ such that $B \notin  \mathscr{Q}_{\sigma}^y.$ Thus, by the Hahn-Banach
    Separation Theorem we further find $X \in \mathscr{R}^{\infty}$ such that
        $$\langle X,B \rangle < 0 < \inf_{A \in \mathscr{Q}_{\sigma}^x} \langle X,A \rangle,$$
    but this implies $\alpha(X) \geqslant y > x \geqslant 0
    $ which contradicts the representation (\ref{eq:repre2}).
    As a consequence $B \in  \mathscr{Q}_{\sigma}^y$ and by the maximality stated in Lemma \ref{lemma:Support}
    we are done.
\end{proof}
\noindent It is worth noting that by choosing $T=0,$ $\mathscr{R}^{\infty}=L^{\infty}(\Omega,\siga_0,\proba)$ and one gets the static AI as in \cite{ChernyMad:2009}.

\section{Consistency with Second Order Stochastic Dominance}\label{sec:Further-Pr1}
The consistence of a performance measure for processes $\alpha$ with the Second order Stochastic Dominance (SSD), hardly depends upon the definition of SSD itself on the spaces $\mathscr{R}^{\infty}.$
To get the equivalent notion in this space of stochastic processes we require (as in the static setting) that if a trade with random cash flow $X \in \mathscr{R}^{\infty}$ has given a greater `utility' than another $Y \in \mathscr{R}^{\infty},$ then it should have a higher performance too, $\alpha(X) \geqslant \alpha(Y).$ Whence, we need to adapt the notion of expected utility in order
to characterize this preference relation via SSD.

\noindent Given a couple of stochastic cash flows $X,Y \in \mathscr{R}^{\infty},$
we recall that $X^{*}=\sup_{t \in [0,T]}|X_{t}|=X^{*}_{+} + X^{*}_{-}$ is the corresponding random variable in $L^{\infty}_+$ and similarly for $Y.$ Therefore, the binary relation defined on $\mathscr{R}^{\infty} \times \mathscr{R}^{\infty}$ by
    $$X \ssd Y \Longleftrightarrow X^{*} \ssd Y^{*},$$
is the analogue of the SSD in the one-time step setting, where what matter are the terminal cash flows. Here instead, we consider the path-dependency
using the running maximum of the reflected (at the origin) processes $ X^{*}$ and $Y^{*}.$ As a consequence, we rewrite the above SSD relation
as
    $$X \ssd Y \Longleftrightarrow \int_{0}^{z} F_{X^{*}}(s) \ud s \leqslant \int_{0}^{z} F_{Y^{*}}(s) \ud s, \qquad \text{for every}\;\, z \in (0, \infty).$$
For the quote on terminal wealth, see \cite[p. 671]{ShaYitz:1994}. In the same paper, we find the characterization of SSD in terms of expected utility. Hence, we define the expected utility on $\mathscr{R}^{\infty}$ in the following sense:
\begin{defi}\label{defi:ExpUt-revi}
For a random cash flow evolving in time $X \in \mathscr{R}^{\infty},$ the version of expected utility for processes over the horizon $[0,T]$ is given by
    $$\mean(U(X)):=\mean(U(X^{*})),$$
where $X^{*} \in L^{\infty}_+,$ and $U: \rio \rightarrow \rio \cup \{\infty\}$ is some concave, non-decreasing function.
\end{defi}
\noindent

\section{Arbitrage and Expectation Consistency}\label{sec:Arbitr-Expec}

In Section \ref{sec:RepreTH} we introduced the family of coherent risk measures for processes
    $$\rho_x(X)=\inf\{m \in \rio \,| \, m \cdot \textbf{1} + X \in \mathscr{A}_x\},\;\; \text{for any} \, X \in \mathscr{R}^{\infty},$$
with the  dual representation
    $$\rho_{x}(X)=\sup_{\pi \in \mathscr{A}^{0}_{x}}\pi(-X),$$
where $\mathscr{A}^{0}_{x}=\{f \in \mathscr{A}^{1} | f(X) \geqslant 0,\;\, \text{for any}\, X \in \mathscr{A}_{x}\}$ is the polar set of $\mathscr{A}_{x}$ in $\mathscr{A}^{1}.$
If we suppose $\mathscr{A}, \mathscr{B}$ are acceptability subsets of $\mathscr{R}^{\infty},$ such that $\mathscr{A} \subset \mathscr{B}$, then $\mathscr{B}^{0} \subset \mathscr{A}^{0}$ holds for the equivalent polar sets in $\mathscr{A}^{1}.$
The above dual representation equals $-\rho_{x}(X) = \inf_{\pi \in \mathscr{A}^{0}_{x}}\pi(X).$ Moreover, $X \in \mathscr{A}_{x}$ implies $-\rho_{x}(X) \geqslant 0$. Therefore, given a family $(\rho_{x})_{x \in \rio_+}$ of coherent risk measures for processes which is monotone with respect to $\mathscr{R}^{\infty}_{+},$ where $\mathscr{R}^{\infty}_{+}$ contains those bounded c\`{a}dl\`{a}g paths $X \geqslant 0,$ this entails $\mathscr{R}^{\infty}_{+} =\mathscr{A}_{0}$ and consequently if $x >0$ then $\mathscr{A}^{0}_{x} \subset \mathscr{A}^{0}_{0}.$ For such AIs arbitrage consistency holds, because
    \begin{eqnarray*}
      \alpha(X) &=& \sup  \left\{x \in \rio_+ \bigg| \inf_{\pi \in \mathscr{A}^{0}_{x}}\pi(X) \geqslant 0  \right\}\\
                &=& \sup\{x \in \rio_+ | -\rho_{x}(X) \geqslant 0\}= \infty.
    \end{eqnarray*}
\noindent Hence, we proved the following:
\begin{theorem}\label{theorem:ArbCons}
    The AI $\alpha: \mathscr{R}^{\infty} \rightarrow [0,\infty]$ defined through a family of monotone coherent risk measures for processes in $\mathscr{R}^{\infty}_{+}$ and the order unit $\textbf{1}$ of $\mathscr{R}^{\infty}$ is Arbitrage Consistent with respect to $\mathscr{R}^{\infty}_{+}.$
\end{theorem}

\noindent Now we come to the extension of the expectation consistency stated in \cite{ChernyMad:2009} for the static case. We again transfer the properties of AIs to the dual system $\left \langle \mathscr{R}^{\infty}, \mathscr{A}^{1} \right \rangle $ from the dual system $\left \langle L^{\infty}, L^{1}\right \rangle.$
\begin{defi}\label{defi:ExpCons1}
    An AI $\alpha: \mathscr{R}^{\infty} \rightarrow [0,\infty]$ is called expectation consistent, if and only if $w(X) >0,$ then $\alpha(X) >0$. The functional $w$ is the one which corresponds to $\textbf{1} \in L^{1}.$
\end{defi}
\begin{prop}
    An AI for processes $\alpha,$ defined on $\mathscr{R}^{\infty}$ is expectation consistent if the level set of zero is $\mathscr{R}^{\infty}_{+}.$
\end{prop}
\begin{proof}
    If the above condition holds, since $\mathscr{A}_{x} \subset \mathscr{A}_{0},$ for every $x \in \rio_+,$
    where as usual $\mathscr{A}_{x}=\{X \in \mathscr{R}^{\infty}| \alpha(X) \geqslant x\}$, we notice that if $w(X) >0$, this implies $\alpha(X) >0.$
\end{proof}

\section{An AI for Processes}\label{sec:Examples}

Assume that $X \in \mathscr{R}^{\infty}$ describes the continuous-time cumulative random return over a finite horizon, and  without loss of generality that the interest rates are zero (avoiding to treat excess returns). We propose to characterize the following AI:
\begin{equation}\label{eq:RAROC}
    \alpha(X):=\frac{\mean(X_T)}{\rho(X)},
\end{equation}
where the denominator represents a coherent risk measure for adapted bounded c\`{a}dl\`{a}g processes, with the convention $\alpha(X) = \infty$ whenever $\rho(X) \leqslant 0.$
The above measure is reminiscent of the SR given by $\frac{\mean(X_T)}{\sd(X_T)},$
where the denominator is the usual standard deviation of the terminal total cumulative return; the numerator measures the expected reward of the underlying investment just at the horizon. Thus, equation (\ref{eq:RAROC}) is a RAROC-type of performance measure provided that the expectation in the numerator is positive. To see why $\alpha(X)$ is an AI for processes, we need to find
the bi-variate process picked from  $\mathscr{Q}_{\sigma}^x$ which is consistent with the representation (\ref{eq:repre1}).
The right choice is the convex combination:
    $$\tilde{A}:=\frac{1}{1+x} B + \frac{x}{1+x} A, \;\;\; x \in \rio_+,\; A,B \in \mathscr{D}_{\sigma}$$
where
    $$B=(B^{\text{pr}}_t, B^{\text{op}}_t)_{t \in [0,T]}:= (0, \ind_{\{u \leqslant t\}})_{t \in [0,T]}.$$
In fact, we have the chain of equivalences for $x>0:$
    \begin{eqnarray*}
      \alpha(X)\geqslant x & \Longleftrightarrow & \frac{\mean(X_T)}{\rho(X)} \geqslant x \\
      & \Longleftrightarrow & \mean(X_T) \geqslant - x \cdot \inf_{A \in \mathscr{D}_{\sigma}} \langle X, A \rangle \\
      & \Longleftrightarrow & \frac{1}{1+x} \, \mean(X_T) + \frac{x}{1+x}\, \inf_{A \in \mathscr{D}_{\sigma}} \langle X, A \rangle \geqslant 0 \\
      & \Longleftrightarrow &  \inf_{A \in \mathscr{D}_{\sigma}} \left[ \frac{1}{1+x} \, \mean(X_T) + \frac{x}{1+x} \langle X, A \rangle  \right] \geqslant 0 \\
      & \Longleftrightarrow &  \inf_{A \in \mathscr{D}_{\sigma}} \left\langle X,\frac{1}{1+x}\cdot B +  \frac{x}{1+x} \cdot A \right\rangle  \geqslant 0 \\
      & \Longleftrightarrow &  \inf_{\tilde{A} \in \mathscr{Q}_{\sigma}^x} \langle X, \tilde{A} \rangle \geqslant 0.
    \end{eqnarray*}
We use Lemma \ref{lemma:Closedness} for the closeness feature of the sets $\mathscr{Q}_{\sigma}^x$ supporting this RAROC-type measure.
In the convex combination defining $\tilde{A},$ the first term $B$ projects the whole random return $X$ onto the terminal date $T$ through the expectation; the second term $A$ pertains to the representation of the coherent measure for processes $\rho(X).$ Note that $\mean(X_T)< \infty$
by the assumption $X \in \mathscr{R}^{\infty}.$
\begin{remark}
The functional $\rho(X):= - \inf_{A \in \mathscr{D}_{\sigma}} \langle X- \mean(X_T), A \rangle$ is obviously a monetary coherent risk measure for $X \in \mathscr{R}^{\infty}.$ By Remark \ref{remark:AccSet} in Section \ref{sec:RepreTH} it induces a reward-risk separation for acceptability at level $x \in \rio_+,$ because $X \in \mathscr{A}_x$ implies the non-negativity of the corresponding $\rho_x(X)$ and thus
    $$\alpha(X) = \frac{\mean(X_T)}{\rho(X)} \geqslant x \;\Longleftrightarrow\;
    \mean(X_T) + \inf_{A \in \mathscr{D}_{\sigma}} \langle X- \mean(X_T), A \rangle \geqslant 0.$$
\end{remark}
\noindent If one chooses the coherent risk measure for bounded paths
   \begin{equation}\label{eq:OurAI}
      \rho(X)=\AVaR_{\gamma}(\inf_{t \in [0,T]}X_t),
   \end{equation}
then the above AI can be made operational. Here $\AVaR_{\gamma}$ is the Average Value-at-Risk at the level $\gamma \in (0,1].$
In fact, a risk measure $\rho(X)$ for bounded processes can be viewed as
    $$\mathscr{R}^{\infty} \xrightarrow{\;\;\; \theta \;\;\;}
            L^{\infty} \xrightarrow{\;\;\; \tilde{\rho} \;\;\;} \rio_+,$$
the composition of a path-transformation with a one-period risk measure $\tilde{\rho}$ applied to the resulting random variable $\theta(X).$ Obviously, $\rho(X)$ would be a coherent monetary risk measure for bounded processes if and only if $\tilde{\rho}$ is a coherent monetary risk measure for single-period cumulative returns, and $\theta$ transforms the paths of $X$ in such a way the properties studied in \cite{CheriEtAl:2004} are preserved. Equivalently, the combined effect of a path-transformation and a static risk measurement is
    \begin{equation}\label{eq:alternRepre}
        \rho(X):=\tilde{\rho}(\theta(X))=-\inf_{Z \in L^1_+,\;\mean(Z)=1}\mean\left[\theta(X) \cdot Z \right].
    \end{equation}
In the current setting, $\theta$ is the running minimum of $X.$ It is important to note that other types of path-transformations can be taken into account, see  \cite[Examples 5.2, 5.5]{CheriEtAl:2004} where $\theta(X)= \frac{1}{T} \int_0^T X_t \ud t$ and it might be viewed as the continuous-time arithmetic average price of the underlying of an Asian option. Anyway, the acceptability of the proposed $\rho(X)$ stems from the coherence of the static $\tilde{\rho}= \AVaR$ together with the monotonicity of the running minimum (properties (1) and (2) are not destroyed). Moreover, the law invariance of $\AVaR$ implies that of $\alpha.$ Clearly, it is also expectation consistent but never arbitrage consistent. Recall that
    $$\AVaR_{\gamma}(\inf_{t \in [0,T]}X_t)= \frac{1}{\gamma} \int_0^{\gamma} \VaR_s(\inf_{t \in [0,T]}X_t) \ud s,$$
where as usual the VaR is defined as the negative of the $s$-quantile of the running minimum's distribution,
    $- \inf\left\{x \in \rio \,| \, \proba(\inf_{t \in [0,T]}X_t \leqslant x) \geqslant s \right\},$
and $\gamma \in (0,1].$

\noindent The widespread CR
    \begin{equation}\label{eq:CalRatio}
        \CR(X) := \frac{\mean(X_T)}{\mean(\sup_{t \in [0,T]} D_t)}, \quad \text{for} \; X \in \mathscr{R}^{p},
       \; \text{and}\;\; p \in [1,\infty],
    \end{equation}
is a classical performance measure depending on the whole investment's path, but fails to be acceptable as we see below.
Here $D=(D_t)_{t \in [0,T]}$ is the drawdown process over $[0,T]$ of the random return $X,$ defined as $D_t:=\sup_{u \in [0,t]} X_u - X_t,$
i.e., it is the drop of $X$ from its running maximum, while the denominator in equation (\ref{eq:CalRatio}) is the
maximum drawdown, i.e. the \emph{greatest} drop of $X$ from its running maximum over the whole horizon (the supremum of $X$ reflected at its running supremum). From now on we assume positive performance indices for processes, whenever $\mean(X_T) > 0$ otherwise the ratio (\ref{eq:CalRatio}) is zero. CR is meant to quantify the expected terminal return of the investment, adjusted by the risk not only at the final date but also including all possible drops from the peaks during the horizon. The easy verification that CR is not an AI for processes is due to the bad behavior of the path-transformation $\theta(X)=\sup_{t \in [0,T]} D_t$ which is not monotone. One can try to replace the expectation $\mean=\tilde{\rho}$ (as the one-period risk functional) with the tail-conditional expectation (i.e. the $\AVaR$ defined on the right tail of the distribution of the maximum drawdown) which is a one-period coherent risk measure, but the lack of monotonicity of the aforementioned $\theta(X)$ destroys acceptability (only convexity is preserved).
\begin{remark}
For two bounded c\`{a}dl\`{a}g cash flows $X,Y \in \mathscr{R}^{\infty}$ such that $X \slaw Y,$ the concept of law invariant AI developed in \cite{ChernyMad:2009} in the static case can be translated int the current setting by limiting ourselves to the case of RAROC-type AIs (\ref{eq:RAROC}). After the corresponding path-transformation is made, the sameness in law of any couple of bounded  c\`{a}dl\`{a}g cash flows then translates to $\theta(X) \slaw \theta(Y),$ thus their transformed paths entail random variables sharing the same probability distribution. We can appeal to \cite[Theorem 5]{ChernyMad:2009}. Firstly, the coherent risk measure for processes $\rho(X)$ in the representation of $\alpha(X)$ can be based on the \emph{weighted VaR}, i.e. the spectral representation $\int_{(0,1]} \AVaR_{\gamma}(\theta(X)) \mu(\ud \gamma)$ with a Borel probability measure $\mu$ on the unit interval. In fact, this in turn is equivalent to the representation of the path-dependent risk measure (\ref{eq:alternRepre}) through a concave distortion, $\tilde{\rho}(\theta(X))=- \int_{\rio} y \, \ud (\Psi_x(F_{\theta(X)}(y))).$  Then, for every $x \in \rio_+$ one defines AI as in the static case by specifying the concave distortion $\Psi_x(u):= \min\{\gamma^{-1} \, u, 1\}$ with $\gamma=1+x$ and $\mu = \delta_{1+x};$ the choice of $Z$ additionally needs $\mean((Z-u)^+) \leqslant \Phi_x(u)$ for every $u \in \rio_+$ such that $\Phi$ is the convex conjugate of the concave distortion. As a by product, $\alpha(X)$ is also consistent with SSD as developed in Section \ref{sec:Further-Pr1}.
\end{remark}

\section{Numerical Comparison}\label{sec:NumExamples}

To give more insight on the behavior of the performance indices discussed in the previous Section, we present here
a simulation exercise based on the following ingredients:
\begin{itemize}
  \item Generation of two L\'{e}vy processes $X,Y$ to describe possible patterns of continuously compounded returns over the horizon
  $[0,T],$ with $T=$1 year;
  \item Determination of six empirical distribution functions $\emdf_n^{X_T},\emdf_n^{\underline{X}},\emdf_n^{\overline{DX}},\emdf_n^{Y_T},\emdf_n^{\underline{Y}},\emdf_n^{\overline{DY}}$
  for the final returns $X,Y$ at $T$ and the corresponding running minimums and maximum drawdowns within $[0,T];$
  \item Estimation of the sample counterparts $\CR_{n},\alpha_{n,\gamma}$ of the CR and the ratio (\ref{eq:OurAI}) presented in Section \ref{sec:Examples}, under the two alternative distributional assumptions;
  \item Comparison of the numerical values deduced from the estimated statistics.
\end{itemize}
\noindent We admittedly carry on the numerical simulation in the unbounded case, albeit as usual the discretization scheme gives bounded sample paths. Thus, modulo any asymptotic consideration we use the results from the simulated paths as an approximation for the bounded case. Due to our main interest in the final step of the above simulation recipe, we choose two parsimonious models of L\'{e}vy processes:
    \begin{enumerate}
      \item A Brownian motion $X$ with constant drift $\mu >0$ modelling the annual expected continuously compounded return, and constant annual volatility $\sigma>0$ of continuous returns;
      \item A Non-normal jump-diffusion $Y$ with the same diffusion part, as suggested by Kou \cite{Kou:2002}.
    \end{enumerate}
\noindent Other models are available for simulating different L\'{e}vy processes such as stable, variance gamma, hyperbolic, stable etc. We rest on the simpler which do not require either the use of special functions nor any series representation. Moreover, we keep to the minimum the difficulty of have no explicit L\'{e}vy measure associated to the processes. Recall that the Kou's model for $Y$ contains four additional parameters: the annual intensity $\lambda$ of a homogeneous Poisson process counting the number of jumps for the non-diffusion part; $\eta_1,\eta_2>0$ such that their reciprocal represent the means of upward and downward deviation from the drift, taken from a double exponential distribution with asymmetric Laplace density; the probability $0<p<1$ of upward jumps. The diffusion part and the jump part are independent, for more details see \cite{Kou:2002}. We set up a routine for the generation of the sample paths of $X$ and $Y$ over a equally-spaced discretization of the horizon by a grid of 1000 time points, through the usual Euler scheme for the corresponding stochastic differential equations with annual $\mu =.15,$ annual $\sigma=.20,$ annual $\lambda=10$ and $\frac{1}{\eta_1}=.02,\frac{1}{\eta_1}=.04.$ This setting entails a mean number of 10 jumps per year with average size $2.2\%$ and jump volatility $4.47\%$ as suggested by \cite{Kou:2002}. The simulation is straightforward because of the independence between the diffusion and the jump part. We compute the final values, the running minimums and the maximum drawdowns of the simulated paths, then we repeat this last step 1000 times as well to get the estimated empirical distribution functions $\hat{\emdf}_n^{X_T},\hat{\emdf}_n^{\underline{X}},\hat{\emdf}_n^{\overline{DX}},$ $\hat{\emdf}_n^{Y_T},\hat{\emdf}_n^{\underline{Y}},\hat{\emdf}_n^{\overline{DY}},$ with $n=1000.$
\begin{table}[!ht]
\begin{center}
\begin{minipage}{100mm}
\caption{Summary Statistics of Simulated Values}  
{\begin{tabular}{lcccccc}
\toprule
   & $\hat{\emdf}_n^{mX}$ & $\hat{\emdf}_n^{DX}$ & $\hat{\emdf}_n^{mY}$ & $\hat{\emdf}_n^{DY}$ & $\hat{\emdf}_n^{X_T}$ & $\hat{\emdf}_n^{Y_T}$ \\
\toprule    
  Skewness       & $-1.6448$ & 1.6036 & $-2.2125$ & 1.5868 & $-0.1415$ & 0.4063\\
  Kurtosis $-3$  & 4.0113    & 4.3742 & 7.1043    & 3.7368 & $-0.0784$ & 0.5070\\
  Median         & $-0.0351$ & 0.0807 & $-0.0283$ & 0.0925 & 0.1212    & 0.3756\\
  Mean           & $-0.0452$ & 0.0879 & $-0.0413$ & 0.1040 & 0.1206    & 0.3946\\
  St. Dev.       & 0.0369    & 0.0328 & 0.0401    & 0.0454 & 0.1006    & 0.2192\\
\toprule    
\end{tabular}}
\label{SummStat}
\end{minipage}
\end{center}
\end{table}
    \begin{figure}\caption{\small Empirical density of the Brownian-return's running minimum.}\label{fig:RunnMinBm}     \centering
        \includegraphics[scale=.45]{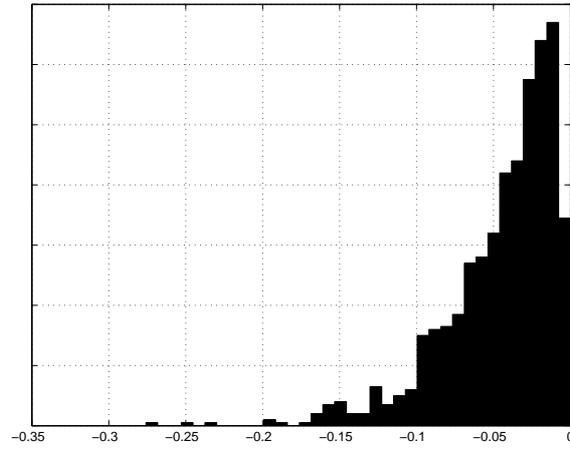}

    \end{figure}
    \begin{figure}\caption{\small Empirical density of the jump-diffusion-return's running minimum.}\label{fig:RunnMinJD} \centering
        \includegraphics[scale=.45]{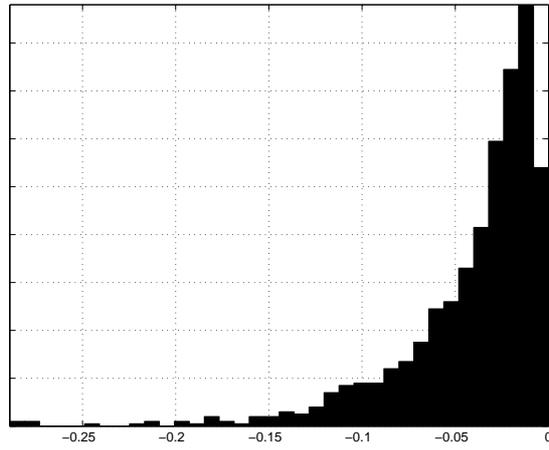}
    \end{figure}
    \begin{figure}\caption{\small Empirical density of the Brownian-return's maximum drawdown.}\label{fig:MaxDDBm} \centering
        \includegraphics[scale=.45]{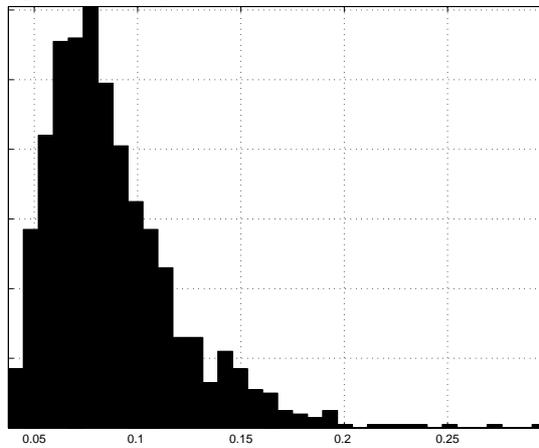}
    \end{figure}
    \begin{figure}\caption{\small Empirical density of the jump-diffusion-return's maximum drawdown.}\label{fig:MaxDDJD} \centering
        \includegraphics[scale=.45]{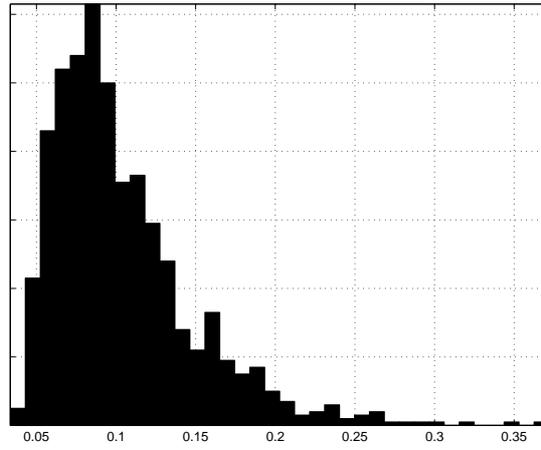}
    \end{figure}
    \begin{figure}\caption{\small Empirical density of the final Brownian return.}\label{fig:FinalBm} \centering
        \includegraphics[scale=.45]{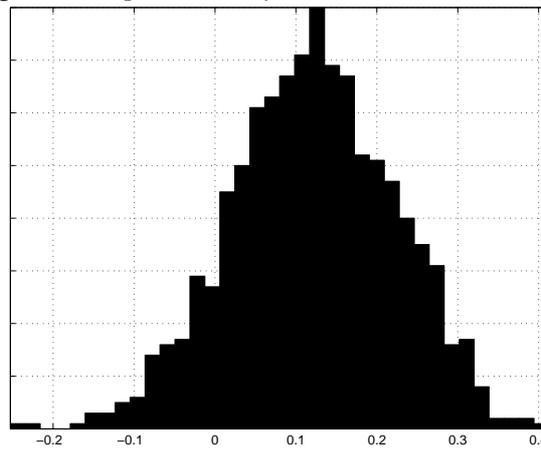}
    \end{figure}
    \begin{figure}\caption{\small Empirical density of the final jump-diffusion return.}\label{fig:FinalJD} \centering
        \includegraphics[scale=.45]{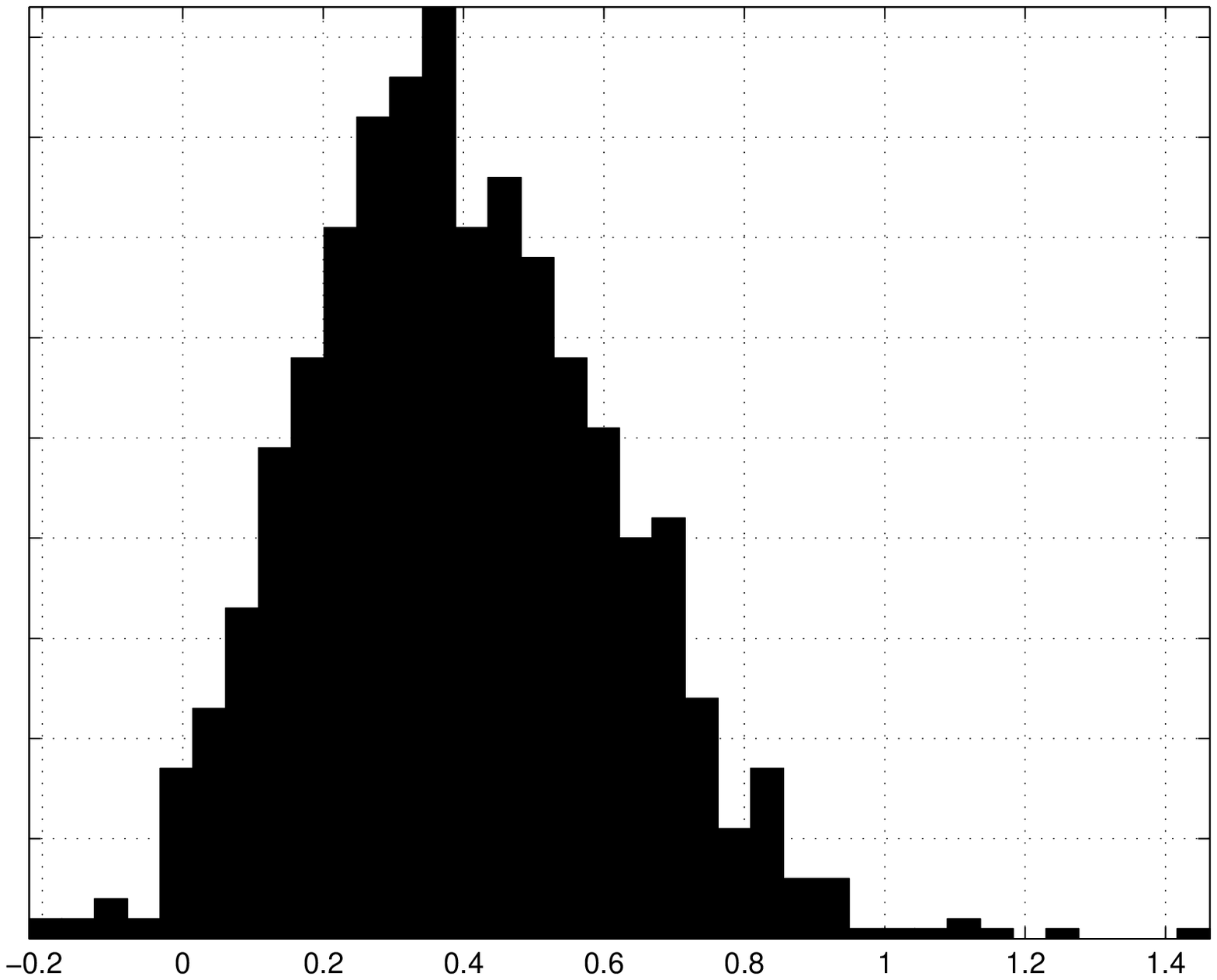}
    \end{figure}
\noindent Analyzing Table \ref{SummStat} together with Figures \ref{fig:RunnMinBm}-\ref{fig:FinalJD}, it is quite evident (as one expects) how the only \emph{quasi} symmetric distribution is that of the simulated final return under the Brownian's law. Indeed, even when this is the assumption one has non-symmetric distributions concerning the running minimum, the maximum drawdown and the one related to the final return under the jump-diffusion's law. The latter put more probability weight to negative (in the case of running minimum) and positive (in the case of maximum drawdown) scenarios of extreme return changes over the horizon. Clearly, these asymmetries affects the ex ante performance measurement. In fact, the Brownian model $X$ is assumed to be a benchmark of normal market values, while the jump-diffusion model $Y$ allows for leptokurtic and asymmetric returns. In presence of shocks, price changes due to good or bad news result in return movements as either overreaction or underreaction of market, according to fat tails and high peak of the jumps distribution.

\noindent Considering the sample $(X_T^i)_{i=1,\ldots,n}$ from the simulation of terminal paths, the sample $(\overline{DX}^i)_{i=1,\ldots,n}$ from the simulation of the maximum drawdown and the sample $(\underline{X}^i)_{i=1,\ldots,n}$ from the simulation of the running minimum, we compute the following estimators in the Brownian case:
\begin{itemize}
  \item $\CR_{n}=\frac{1/n\sum_{i=1}^n X_T^i}{1/n\sum_{i=1}^n \overline{DX}^i};$
  \item $\alpha_{n,\gamma}=\frac{1/n\sum_{i=1}^n X_T^i}{1/k\sum_{i=k}^n \underline{X}^{(k)}},$ with $k:= [n\gamma]$ the greatest integer less than or equal to $n\gamma$ and $(\underline{X}^{(i)})_{i=1,\ldots,n}$ being the corresponding ordered sample; we set $\gamma=.01$ and $\gamma=.05;$
  \item $\SR_n^B=\frac{1/n\sum_{i=1}^n X_T^i}{\left(\frac{1}{n}\sum_{i=1}^n \underline{X}_i^2- \left(\frac{1}{n}\sum_{i=1}^n \underline{X}_i \right)^2 \right)^{1/2}};$
  \item $\SR_n^J=\frac{1/n\sum_{i=1}^n X_T^i}{\left(\frac{1}{n}\sum_{i=1}^n \overline{DX}_i^2- \left(\frac{1}{n}\sum_{i=1}^n \overline{DX}_i \right)^2 \right)^{1/2}}.$
\end{itemize}
\noindent The last two estimators refer to SR where the standard deviation is based on the running minimum and the maximum drawdown, respectively. The estimators in the case of the jump-diffusion $Y$ are obtained in the same way. All the proposed performance indices provide a reward measurement in terms of final expected return.
\begin{table}[!ht]
\begin{center}
\begin{minipage}{100mm}
\caption{Simulated Values of Performance Indices}
{\begin{tabular}{lccc}
\toprule
   &  & Brownian motion  & Jump-diffusion \\
\toprule    
Calmar ratio   & $\hat{\CR}_{n}$ & 1.3718 & 2.8133\\
$\gamma =.05$  & $\hat{\alpha}_{n,\gamma}$ & 0.7950 & 1.3130\\
$\gamma =.01$  & $\hat{\alpha}_{n,\gamma}$ & 0.5998 & 0.9182\\
Estimate for $\SR^B$   & $\frac{\mean(\text{final return})}{\sd(\text{running minimum})}$   & 3.2652 & 5.2582\\
Estimate for $\SR^J$  & $\frac{\mean(\text{final return})}{\sd(\text{maximum drawdown})}$  & 3.6670 & 5.2421
\end{tabular}}
\label{SimRatios1}
\end{minipage}
\end{center}
\end{table}
\noindent CR and its relatives (viz. the last tow rows of Table \ref{SimRatios1}) entail overestimation of financial performance, mainly due to the use of either the expectation or the standard deviation to build path-dependent risk measures which suffer from the aforementioned asymmetries (skewness, kurtosis and fat-tails). Our performance indices $\alpha_{\gamma}$ have smaller values. This maybe suggests a more prudent performance evaluation by taking into account adverse downside scenarios under the proper distributional assumption towards a forecast of phenomenon like margin calls, rebalancing of trading positions, counter-party risks, fund redemption. In this perspective, the index based on (\ref{eq:RAROC}) with coherent monetary risk measure for processes being the one proposed in Section \ref{sec:Examples}, gives a sensible trade-off between acceptability as a theoretical feature and the ability to capture extremes changes in the return profile during the whole holding period. Furthermore, the resulting measurement is compatible with the heightened need for performance-tracking tools that embed economic cost of risk, and the consequential use of the RAROC by financial institutions.
\begin{remark}
    When random cash flows are asymmetric and fat-tailed, performance evaluation using SR-type indices is questionable. Investors' preference for certain portfolio compositions heavily relies on more than the first two moments of the cash flow's distribution, and clear enough they could be wealth-seeking and risk-adverse but prefer portfolios with higher volatility. In these situations, the $\AVaR$ behaves better as a static risk measure and assuming a tangential portfolio (in the sense of minimum $\AVaR$ for a fixed tail probability $\gamma$) it constitute a sensible substitute for the volatility given by the standard deviation or the alike, when portfolio optimization comes into play and asset allocation or ranking of investment funds concerns can be tackled more efficiently. Observe in addition how all the estimators above are biased for finite samples but have asymptotic efficiency as $n$ grows more and more.
\end{remark}

\section{Conclusions}\label{sec:Concl}

Intra-horizon risk measurement has recently gained increasing consideration among academics and practitioners. For example, VaR may be improved in the time dimension by focusing on the return's distribution within the whole investment horizon. On the fund management side, performance analysis with drawdown is well developed and the evaluation and even the ranking of investment portfolios by measures such as CR is popular among fund managers. These indices are based on the risk-adjustment of returns by recording all the information about the evolution of processes modelling profit and losses within a fixed horizon. Motivated by these well documented facts and in addition by the growing use of the RAROC as a tool of banks' backward-looking manner at business performance-tracking and risk management, we propose a performance measurement for processes. Therefore, we look on one hand at the approach of Cherny \& Madan (2009) to static performance evaluation driven by the concept of AIs, and on the other hand to the work of Cheridito et al. (2004) which generalizes static coherent risk measures to monetary risk measures for processes. Our main contribution is to extend the representation of an AI as applied not to random variables giving terminal cash flow, but instead to stochastic processes modelling the evolution of cash flows over a finite horizon. The domain of an AI is now the class of bounded c\`{a}dl\`{a}g processes, with new representation results. Our contribution does not overlap that of Bielecki et al. (2014), since they treat dynamic coherent AIs with a focus on time-consistency, while we develop static AIs for processes. Notwithstanding, our framework can embody information about the whole investment process even if sequential conditioning is ruled out. We propose an AI expressed as the ratio of expected terminal cash flow to a coherent monetary risk measure for processes. Eventually, we compare different numerical values of this acceptability ratio and other performance measures, to appreciate the embedding of information about stressed scenarios concerning the whole horizon. Our ratio could be a sensible compromise between acceptability (as a desirable aggregate property) and dependency on the whole cash flow's path (as the quest of practitioners). This is also compatible with the widespread use of RAROC and fixes the lack of acceptability of CR. The research agenda on this topic will include the generalization of our AI to the unbounded case (e.g. the space $\mathscr{R}^1$), as well as a more comprehensive definition of law invariance of AIs for processes. Also a study of the asymptotic behavior of AIs when the horizon becomes infinite is desirable.

\appendix

\section{Appendix: Duality Relations for Processes}\label{appx}

The main results of Section \ref{sec:RepreTH} are based on the generalization of acceptability from spaces of random variables to spaces of stochastic processes, partly introduced in Section \ref{sec:AIoP}. To keep the present paper as much self-contained as possible, we list in this Appendix some facts about these spaces of processes and the corresponding duality relations following closely \cite{CheriEtAl:2004}. For $p \in [1,\infty]$ the collection
    $$\mathscr{R}^p := \left\{ X:[0,T]\times \Omega \to \rio\begin{array}{l|l}
                                      & X \; \text{c\`{a}dl\`{a}g}\\
                                      & (\mathscr{F}_t)\text{-adapted} \\
                                      & \left\| X\right\|_{\mathscr{R}^p} < \infty
                                  \end{array}
     \right\}$$
is a Banach space. Recall that increasing processes $A : [0,T]\times \Omega \to \rio$ (i.e. adapted, with positive right-continuous and increasing paths) induces a measure $\ud A_t(\omega).$ In case $A$ has right-continuous paths with finite variation, its unique decomposition $A = A^+ - A^-$ into two right-continuous increasing processes induces $\proba$-a.s. positive measures on $[0,T]$ with disjoint support. The total variation of such process is the random variable $\text{Var}(A):= A^+_T+A^-_T.$ Moreover, if $A$ is optional (i.e. a measurable on $[0,T]\times \Omega$ equipped with the $\sigma$-algebra generated by the adapted c\`{a}dl\`{a}g processes) then $A^+,A^-$ are optional. When $A$ is predictable (i.e. measurable on $[0,T]\times \Omega$ equipped with the $\sigma$-algebra generated by the adapted continuous processes) then $A^+,A^-$ are predictable. Thus, for $q \in [1,\infty]$ we have the collection
$$\mathscr{A}^q := \left\{ A:[0,T]\times \Omega \to \rio^2 \begin{array}{l|l}
                                      & A =(A^{\text{pr}},A^{\text{op}})\; \text{right-continuous, finite variation}\\
                                      & A^{\text{pr}}\;\text{predictable,}\;A_0^{\text{pr}}=0 \\
                                      & A^{\text{op}} \; \text{optional, purely discontinuous} \\
                                      & \text{Var}(A^{\text{pr}}) + \text{Var}(A^{\text{op}}) \in L^q
                                  \end{array}
     \right\}$$
This collection equipped with the norm $\norm{A}_{\mathscr{A}^q}:=\norm{\text{Var}(A^{\text{pr}}) + \text{Var}(A^{\text{op}})}_q$ is a Banach space. The subset $\mathscr{A}^q_+$ containing those $A \in \mathscr{A}^q$ with the predictable and optional parts being non-negative and increasing. The bilinear form
    $$\langle X,A \rangle := \mean\left[ \int_{(0,T]} X_{t^-}\ud A^{\text{pr}}_t + \int_{[0,T]} X_t\ud A^{\text{op}}_t \right]$$
defined on $\mathscr{R}^p \times \mathscr{A}^q$ for $p,q \in [1, \infty]$ such that $p^{-1}+q^{-1}=1,$ is also continuous and put in duality the spaces $\mathscr{R}^p$ and $\mathscr{A}^q.$ Indeed, it is well known that $|\langle X,A \rangle| \leqslant \left\| X\right\|_{\mathscr{R}^p} \left\| A\right\|_{\mathscr{A}^q}.$

\section{Appendix: Coherent Monetary Risk Measures for Processes}\label{appx2}

The apparatus introduced in Appendix \ref{appx}, is employed for extending the  structure theorem for coherent risk measures
$\rho : L^{\infty}  \to \rio:$ Given the $\sigma(L^{\infty},L^1)$-closed acceptance set $\mathscr{C}=\{X \in L^{\infty} | \, \rho(X) \leqslant 0\}$ for a discounted terminal cash flow $X,$ then
    $$\rho(X)= - \inf_{Q \in \mathscr{D}} \mean^{Q}(X) = - \inf_{Z \in L_+^1,\; \mean(Z)=1} \mean(X \cdot Z),$$
for a certain set $\mathscr{D}$ of probability measures absolutely continuous with respect to $\proba,$ with corresponding Radon-Nikodym derivatives $Z= \frac{\ud Q}{\ud \proba}.$ The typical proof of this result uses $\mathscr{C}$ to support the representation itself, see \cite{Delbaen:2012} for a thorough treatment of static (also dynamic) risk measures based of the corresponding coherent monetary utility functional $\phi=-\rho.$
Now, static coherent risk measures are in duality with static AIs since
    $$\alpha(X)= \sup\{x \in \rio_+ | \, \rho_x(X) \leqslant 0\},$$
where $\rho_x$ is an indexed family of coherent risk measures whose representation is supported by an $x$-increasing family $(\mathscr{D}_x)_{x \in \rio_+}$ of absolutely continuous probability measures,  together with the corresponding Radon-Nikodym derivatives and acceptance sets. In the present paper we do a similar construction by indexing the set
    $$\mathscr{D}_{\sigma}:= \left\{ A \in \mathscr{A}_+^1 \big| \; \norm{A}_{\mathscr{A}^1}=1\right\},$$
and then working with the bilinear form $\langle X,A \rangle$ rather than the classical expectation $\mean(X \cdot Z).$ Obviously, in this extended framework we have that
    $$\rho(X)= -\inf_{A \in \mathscr{D}_{\sigma}} \langle X,A \rangle, \;\;\; X \in \mathscr{R}^{\infty},$$
with corresponding $\sigma(\mathscr{R}^{\infty},\mathscr{A}^1)$-closed acceptance set $\mathscr{C}=\{X \in \mathscr{R}^{\infty} | \, \rho(X) \leqslant 0\},$ see \cite[Corollay 3.5]{CheriEtAl:2004}. In our paper we clearly attach a numerical acceptability level $x \in \rio_+$ to the supporting set $\mathscr{D}_{\sigma}$ to entail the duality with AIs for bounded c\`{a}dl\`{a}g processes. Finally, a few facts to note. First, the space $\mathscr{R}^{\infty}$ is invariant with respect to the probability measure $\proba$ or its equivalents. Second, for
$p,q \in [0,\infty]$ such that $p^{-1}+q^{-1}=1$ the space $\mathscr{A}^q$ can be identified with the topological dual $(\mathscr{R}^p)^*$ of the space $\mathscr{R}^p,$ while $\mathscr{A}^1 \subset (\mathscr{R}^{\infty})^*.$ Thus \cite[Theorem 3.3, Corollary 3.5]{CheriEtAl:2004} provide those coherent monetary utility functionals on $\mathscr{R}^{\infty}$ that can be represented with vectors of $\mathscr{A}^1.$

\noindent \textbf{Acknowledgements}\\
\noindent The authors are grateful to two anonymous Referees and an Associate Editor for
their helpful suggestions and critical comments, that help us to improve the presentation of the paper.


\end{document}